%% file: main.tex
\documentclass[runningheads, 11pt]{llncs}

\usepackage[margin=1in]{geometry}

\input{macros}

\usepackage{amssymb}

\begin{document}
\title{Computing Equilibria of Prediction Markets via Persuasion}

\author{Jerry Anunrojwong\inst{1} \and
  Yiling Chen\inst{2} \and
  Bo Waggoner\inst{3} \and
  Haifeng Xu\inst{4} }
\authorrunning{Anunrojwong, Chen, Waggoner and Xu}


\institute{Columbia Business School \and
  Harvard University \and
  University of Colorado Boulder \and
  University of Virginia}

\maketitle

\begin{abstract}
  We study the computation of equilibria in prediction markets in perhaps the most fundamental special case with two players and three trading opportunities.
  To do so, we show equivalence of prediction market equilibria with those of a simpler signaling game with commitment introduced by Kong and Schoenebeck (2018).
  We then extend their results by giving computationally efficient algorithms for additional parameter regimes.
  Our approach leverages a new connection between prediction markets and Bayesian persuasion, which also reveals interesting conceptual insights. 
\end{abstract}


\input{intro}

\input{prelim}
\input{equiv}

\input{algs}

\input{conclusion}


\bibliographystyle{splncs04}
\bibliography{references,citations}

\newpage

\appendix

\input{persuasion-general}
\input{omitted-proofs}

%

\end{document}

%% file: macros.tex
\usepackage{amsmath,amsfonts,amssymb}
\usepackage{xcolor}

\newcommand{\citep}{\cite}
\newcommand{\citet}{\cite}

\newcommand{\authcomment}[3]{}  

\newtheorem{fact}{Fact}
\newenvironment{proofsketch}{%
  \proof}{\endproof}

\DeclareMathOperator*{\Ex}{\mathbb{E}}

\newcommand{\RR}{\mathbb{R}}

\newcommand{\inprod}[2]{\left\langle #1 , #2 \right\rangle}  
\renewcommand{\vec}[1]{\mathbf{#1}}  

\newcommand{\E}{\mathcal{E}}  
\newcommand{\A}{\mathcal{A}}  
\newcommand{\B}{\mathcal{B}}  
\renewcommand{\S}{\mathcal{S}}  

\newcommand{\post}[1]{\vec{p}_{#1}}  

\DeclareMathOperator{\poly}{poly}

\DeclareMathOperator*{\argmax}{arg\,max}

%% file: intro.tex
\section{Introduction} \label{sec:intro}

Prediction markets allow participants to buy and sell financial contracts whose payoff is contingent on the outcome of a future event.
The market aggregates these decisions, which reveal beliefs about the event, into a collective prediction.
Researchers study their game-theoretic properties to understand how these markets function in practice as well as how to better design them to encourage information elicitation and aggregation.


The widely-studied \emph{scoring-rule based markets (SRM)}~\citep{hanson2003combinatorial} utilize \emph{proper scoring rules} $R(\vec{p},e)$, which assign a score to each prediction $\vec{p}$ on any given outcome $e$ of the event.
Each participant $t=1,\dots,T$ arrives and updates the market prediction from $\vec{p}^{t-1}$ to $\vec{p}^t$, and receives a payoff of her improvement in score, $R(\vec{p}^t,e) - R(\vec{p}^{t-1},e)$, after the event outcome $e$ is revealed.

Despite the apparent simplicity of this game, its equilibria have been challenging to describe.
We have two primary motivations for doing so.
First, prediction markets are popular in practice, and understanding the properties of their equilibria may be helpful in determining how to design such markets.
Second, the SRM is a very simple but apparently deep extensive-form signaling game.
Understanding it may lead to general insights regarding value of information and connections to other signaling settings.
Therefore, this paper seeks algorithms and characterizations that further our understanding of these games.

\vspace{3mm}
\noindent {\bf The Alice-Bob-Alice (ABA) game and prior work.}
Historically, equilibria of markets have proven difficult to describe even in the special but perhaps the most fundamental ``Alice-Bob-Alice'' (ABA) case.
Here there are only two players and three trading opportunities.
Alice observes a private signal from a set $\A$ while Bob receives a private signal from a set $\B$.
They can be correlated with each other and with the (random) event being predicted, which has outcomes drawn from a set $\E$.
Alice, participating at $t=1$, can choose to predict truthfully, withhold information, or even bluff and make a knowingly false prediction.
This might mislead Bob into a poor prediction at $t=2$, leaving Alice the opportunity to improve the market score significantly at $t=3$.

A sequence of works \citep{chen2007bluffing,dimitrov2008non,chen2010gaming,gao2013jointly} focused on the popular log scoring rule and found conditions under which Alice fully reveals all information in stage $1$ as well as cases where she reveals no information.
Chen and Waggoner \citet{chen2016informational} generalized these results to a characterization of pairs (players' signals, scoring rule) under which the first player is always truthful (termed \emph{informational substitutes}) or withholds all information (\emph{informational complements}).
All of the results mentioned so far extend to general prediction markets with any number of players, yet solving the Alice-Bob-Alice case was often the key step.

However, one major open problem left in \citet{chen2016informational} is the computational tractability of  determining whether players' signals satisfy the substitutes condition, complements condition, or neither.
The aforementioned papers also leave open what happens in the ``neither'' case, i.e. when Alice uses some nontrivial strategy in the first stage. To our knowledge, 
Kong and Schoenebeck \citet{kong2018optimizing} are the first to address these questions.
It introduced a signaling game, the \emph{Alice-Bob-Alice game with commitment}, that simplifies some aspects of prediction markets from an analysis perspective.
Payoffs are defined as in the Alice-Bob-Alice SRM above.
But instead of directly making a prediction in round $1$, Alice reports according to some signaling scheme conditioned on her private information. 
Bob observes Alice's signal and Alice is assigned $\vec{p}^1 = $ the posterior event distribution conditioned on this signal.
Crucially, Alice must commit to this signaling scheme and it is known to Bob in advance, so she cannot bluff or mislead him by deviating to another signal or prediction.
For this game, \citet{kong2018optimizing} gave a fully polynomial-time approximation scheme (FPTAS) for computing an optimal signaling scheme of Alice when the number of possible realizations of Alice's private information, $|\A|$, is constant, and the scoring rule satisfies a rather strong separability and smoothness condition.

\vspace{3mm}
\noindent{\bf Our Results.} Our first result establishes a formal connection between ABA game with and without commitment. We prove that Alice's optimal commitment in the ABA game is also (up to negligible $\epsilon$) part of an equilibrium in the corresponding prediction market (without commitment). 
This shows, perhaps surprisingly, that any equilibrium that can be achieved when Alice is forced to commit to a signaling scheme can also be achieved in a market without commitment or explicit signaling.
In other words, finding equilibria in prediction markets reduces to a pure signaling problem. 


Given this characterization, we then focus our attention on designing algorithms for the ABA game with commitment. Here, we extend the results of \citet{kong2018optimizing} to several other cases, although we do not solve the Alice-Bob-Alice game in full generality.  Our results are built upon an interesting connection between Alice's signaling problem and \emph{Bayesian persuasion} \citep{Kamenica2011,kolotilin_persuasion_private} --- in some sense, Alice's signaling scheme in round $1$ is ``persuading'' Bob to make certain reports.  We formalize this connection by proving that Alice's signaling problem reduces to Bayesian persuasion of a \emph{privately informed} receiver, but with a persuasion objective that is specific to prediction markets. 
As a direct application of this connection, we exhibit an efficient and exact algorithm for Alice's optimal signaling in the case $|\B| = O(1)$ but under the assumption that the expected scoring function is piece-wise linear with polynomially many  pieces. Though this restriction appears restrictive, we hope this result may serve as a stepping stone to future work. Next, we leverage techniques from algorithmic persuasion to design an FPTAS for the case   $|\A| = O(1)$ under a natural smoothness assumption on the scoring function. This results strictly generalizes --- and interestingly, also much simplifies --- the main result of Kong and and Schoenebeck \citet{kong2018optimizing}. Finally, to show the generality of our technique, we use a similar idea to design an FPTAS for the case that both $|\B|,|\E| = O(1)$.

%% file: prelim.tex
\section{Preliminaries} \label{sec:prelim}

\subsection{Signals and probabilities}
A \emph{signal} is a random variable, denoted by a capital letter, taking values in an \emph{outcome space} written in calligraphics.
In particular, there are four signals of interest in this paper: $E$, $A$, $B$, and $S$.
The signal $E$ is a future event we would like to predict having a finite set of outcomes $\E$.
The goal of a prediction market is to elicit forecasts about $E$ in the form of probability distributions in $\Delta({\E})$, the probability simplex over $\E$.
For an outcome $e \in \E$, we write $\Pr[e]$ as shorthand for $\Pr[E=e]$, and so on for the other signals.

In this paper, there will always be two players, Alice and Bob.
Alice observes a signal $A$ with finite outcome space $\A$, while Bob observes $B$ in the finite space $\B$.
There is a prior distribution $\mu(e,a,b)$ on the joint realizations of $e \in \E$, $a \in \A$, and $b \in \B$. The prior distribution is common knowledge to Alice and Bob.
Alice will be choosing to send a signal $S$ in space $\S$.
A \emph{signaling scheme} is represented as a function $\pi: \S \times \A \to [0,1]$ where $\pi(s,a) = \Pr[S=s,A=a]$ such that $\pi$ satisfies $\sum_{s \in \S} \pi(s,a) = \sum_{e,b}\mu(e,a,b)$ for all $a \in \A$. 

%

\subsection{Prediction market model}

\paragraph{Proper scoring rules.}
A scoring rule is a function $R: \Delta({\E}) \times \E \to \RR \cup \{-\infty\}$ that assigns a score $R(\vec{w},e)$ to the prediction $\vec{w}$ when the event $E$ of our interest is realized to $e$. 
We write $R(\vec{w}';\vec{w}) = \Ex_{E\sim\vec{w}} R(\vec{w}', E)$ for the expected score of prediction $\vec{w}'$ when $E$ is drawn from $\vec{w}$.
It is \emph{strictly proper} if for all $\vec{w} \neq \vec{w}'$, $R(\vec{w}';\vec{w}) < R(\vec{w};\vec{w})$.
That is, for any belief $\vec{w}$, one uniquely maximizes expected score by reporting $\vec{w}$.
We rely on the following characterization.
\begin{proposition}[\citet{mccarthy1956measures,savage1971elicitation,gneiting2007strictly}]
  For every strictly proper scoring rule $R$, there exists a strictly convex function $G: \Delta({\E}) \to \mathbb{R}$ such that $R(\vec{w};\vec{w}) = G(\vec{w})$.
  Conversely, from every strictly convex $G$, one can construct a strictly proper scoring rule $R$ such that $G(\vec{w}) = R(\vec{w};\vec{w})$.
\end{proposition}
\begin{example}
  The log scoring rule is defined as $R(\vec{w},e) = \log w_e$, i.e. the logarithm of the probability assigned to $e$.
  Its ``expected score function'' is $G(\vec{w}) = \sum_e w_e \log w_e = -H(\vec{w})$, the negative of Shannon entropy.
  The quadratic scoring rule is $R(\vec{w},e) = 2w_e - \|\vec{w}\|_2^2$.
  Its expected score function is $G(\vec{w}) = \|\vec{w}\|_2^2$.
  Both are strictly proper.
\end{example}


\paragraph{Automated prediction market.}
In this paper we focus on the popular \emph{automated scoring-rule market (SRM)} framework of \citet{hanson2003combinatorial}.
The market is parameterized by a finite set of event outcomes $\E$, a strictly proper scoring rule $R$, and an initial prediction $\vec{p}^{0} \in \Delta({\E})$.
The participants arrive in a fixed, predefined order.
Each round $t=1,\dots,T$, the arriving participant observes the previous prediction $\vec{p}^{t-1}$ and replaces it with a prediction $\vec{p}^{t}$.
At the end, the event outcome $E=e$ is observed and the arriving participant at time $t$ is paid
\begin{equation}
  R(\vec{p}^{t},e) - R(\vec{p}^{t-1},e) . \label{eqn:market-scoring-rule}
\end{equation}
One of the key properties this payoff rule inherits from $R$ is ``one-step'' truthfulness:
\begin{fact} \label{fact:one-step-truthful}
  If every player arrives only once, then it is a strictly dominant strategy to set $\vec{p}^t$ to the player's true posterior belief conditioned on all information they have observed.
\end{fact} 
This follows immediately because $R$ is a proper scoring rule and the second term in (\ref{eqn:market-scoring-rule}) is not under the player's control.

However, if players participate multiple times, it might be beneficial to withhold information (or possibly even bluff).
This motivates study of the \emph{Alice-Bob-Alice (ABA) market}, a prediction market with two players and three rounds where Alice participates in rounds $1$ and $3$ while Bob participates in round $2$.
Despite its apparent simplicity, this special case captures many of the challenges of general markets and has been studied in e.g. \citet{chen2010gaming,gao2013jointly,kong2018optimizing}.

\paragraph{Equilibrium in markets.}
In the prediction market game, a strategy for Alice consists of a pair of possibly-randomized functions $\sigma_1,\sigma_3$ defining her predictions at rounds $1$ and $3$.
We have $\sigma_1: \A \to \Delta({\E})$, i.e. Alice plays $\vec{p}^1 = \sigma_1(A)$.
Next, $\sigma_3: \A \times \Delta({\E}) \times \Delta({\E}) \to \Delta({\E})$, where Alice at round $3$ plays $\vec{p}^3 = \sigma_3(A, \vec{p}^1, \vec{p}^2)$.
Similarly, a strategy for Bob is a possibly-randomized function $\sigma_2 : \B \times \Delta({\E}) \to \Delta({\E})$ where he plays $\vec{p}^2 = \sigma_2(B, \vec{p}^1)$.

For $t\in\{1,2,3\}$, define the expected net score for the prediction at round $t$ to be
  \[ u_{t}((\sigma_1,\sigma_3),\sigma_2) = \Ex_{A,B,E,\sigma_1,\sigma_2,\sigma_3} \left[R(\vec{p}^t,E) - R(\vec{p}^{t-1},E) \right] . \]
Alice's total expected utility is $u_A((\sigma_1,\sigma_3),\sigma_2) := u_{1} + u_{3} .$
Similarly, Bob's expected utility is $u_B((\sigma_1,\sigma_3),\sigma_2) := u_2$.

A set of strategies $((\sigma_1,\sigma_3),\sigma_2)$ are a \emph{Bayes-Nash equilibrium (BNE)} if each is a best response to the other, i.e. for all $(\sigma_1',\sigma_3')$, $u_A((\sigma_1',\sigma_3'),\sigma_2) \leq u_A((\sigma_1,\sigma_3),\sigma_2)$, and similarly for all $\sigma_2'$, $u_B((\sigma_1,\sigma_3),\sigma_2') \leq u_B((\sigma_1,\sigma_3),\sigma_2)$.

In extensive-form games such as prediction markets, BNE can include ``non-credible'' threats.
For example perhaps in BNE, Bob may threaten to reveal no information in the second round if Alice deviates from the equilibrium strategy.
This is not credible because, if Alice were to actually deviate, Bob's best response would still be to predict truthfully according to his beliefs.
Therefore, in this paper we focus on \emph{perfect Bayesian equilibrium (PBE)}.
Informally, a BNE $((\sigma_1,\sigma_3),\sigma_2)$ is a PBE if, off the equilibrium path, these strategies still best-respond according to some beliefs that are consistent with Bayesian updating on the player's own signal and some information about their opponent's signal. See the full version for a formal definition.

\subsection{ABA game with commitment}
Although prediction market equilibria generally capture relative value of information, there are several technical complications.
First, in principle it could be that a prediction of Alice's does not reveal her signal for the coincidental reason that two signals give the same posterior belief.
For example, in the case where both players receive a uniformly random bit and $E = A \oplus B$ (the XOR), Alice's posterior on $E$ is uniformly random regardless of which signal she receives.
Second is the question of \emph{commitment}.
It might be that equilibria of prediction markets do not completely reflect the relative value of information and idealized signaling schemes because Alice is unable to commit to such a scheme.

This motivates us to study the more mathematically clean \emph{ABA game with commitment}. 
Introduced in \citet{kong2018optimizing}, this ``game'' can be phrased as a single-player decision problem, fully specified  by $\{  G,\mu  \}$ where: 
convex function $G: \Delta({\E}) \to \RR \cup \{-\infty\}$ is chosen by the designer; 
$\mu$ is the prior on $(A,B,E)$.
Alice makes the only decision in the game by selecting a signaling scheme $\pi: \S \times \A \to [0,1]$.
This signaling scheme is announced to Bob.
Nature draws $(A,B,E) \sim \mu$ and draws $S \sim \pi(\cdot \mid A)$.
Bob observes the signal $S$, updates to a posterior $\vec{p}_{S,B}$, and receives utility $R(\vec{p}_{S,B},E) - R(\vec{p}_S,E)$. Then Alice receives utility $R(\vec{p}_S,E) - R(\vec{p},E) + R(\vec{p}_{A,B},E) - R(\vec{p}_{S,B},E)$ in total.
Crucially, this payoff structure makes the game constant-sum since for each $A=a,B=b,E=e$, the sum of Alice's and Bob's utilities equals $R(\vec{p}_{a,b},e) - R(\vec{p},e)$, which is fixed.\footnote{This is a slight departure from the formalization of the game in \citet{kong2018optimizing}. There, Alice did not automatically observe Bob's signal, causing complications in the case where Bob's report $\vec{p}_{S,B}$ could be the same for two different outcomes $b,b' \in \B$.}

The interpretation of these payoffs is that Alice comes to the prediction market, announces signal $S$, and predicts the posterior conditioned on $S$.
Then, Bob arrives, sees $S$, 
announces $B$, and predicts the posterior conditioned on both $S$ and $B$ (via Bayesian update).
Finally, Alice arrives, 
announces $A$, and predicts the posterior given both $A$ and $B$.
In other words, as phrased by \cite{howard1966information,chen2016informational}, Alice receives the \emph{marginal value} of signal $S$ over the prior; then Bob receives the marginal value of $B$ over $S$; and finally, Alice receives the marginal value of $A$ over $S,B$.

\subsection{Bayesian Persuasion}\label{sec:prelim:BP}
The ABA game turns out to be relevant to the Bayesian persuasion model. A persuasion game is played between a  \emph{sender} and a \emph{receiver}. The receiver is faced with selecting an action $i$ from  $[k] = \{ 1,\cdots,k \}$. Both the sender and receiver utility depend on the receiver's action as well as a state of nature $e$ supported on $\E$. Formally, the sender and receiver payoff function are $v(i,e)$  and $u(i,e)$ where $i \in [k]$ and $e \in \E$. 

Particularly relevant to this work is the model of \emph{Bayesian persuasion with a privately informed receiver}, first studied by Kolotilin \emph{et al} \cite{kolotilin_persuasion_private}. Here, the sender and receiver each observe a private signal regarding the state of nature $E$, which may be correlated with each other. Let $A \in \A$ and $B \in \B$ denote the (random) signal observed by the sender and receiver, respectively. The joint distribution of $A,B,E$ is public knowledge and denoted as  $\mu(e,a,b)$. The Bayesian persuasion model studies how the sender can maximize her expected utility by \emph{committing} to a signaling scheme $\pi:  \S \times \A \to [0,1]$ to strategically influence the receiver's belief about $e$ and consequently his optimal action.\footnote{Such a signaling scheme is also called an \emph{experiment} by Kolotilin \emph{et al} \cite{kolotilin_persuasion_private}. We remark that their model is  a special case of the general model we described here, with independent $A,B$ and binary receiver actions.} Here, again, $\S$ is the set of signal outcomes. 
In Section \ref{sec:aba-bp}, we will formalize the connection to prediction markets, which involves Alice ``persuading'' Bob to make certain reports but with a particular form of sender objectives specific to prediction markets.  

%% file: equiv.tex
\section{Equivalence with and without Commitment} \label{sec:equiv}

In this section, we show that Alice's optimal signaling scheme in the ABA game with commitment yields an approximate PBE in the Alice-Bob-Alice prediction market (without commitment).
Thus, we can next focus on solving the ABA game with commitment.
In this section, to simplify technicalities, we assume that the proper scoring rule $R$ has a \emph{differentiable} convex expected score function $G$.

First, we formalize the sense in which Alice uses a signaling scheme even in a prediction market.
This perspective has appeared in prior works on equilibria of markets, though a precise result may not have been stated. 
Informally, it says that in \emph{any} equilibrium, Alice's equilibrium strategy can be written as reporting the posterior conditioned on a signal she draws from a private scheme.
Recall from Fact \ref{fact:one-step-truthful} that, because Bob only participates once and the market uses a strictly proper scoring rule, his unique best response is always to report truthfully according to his information and beliefs.
\begin{lemma} \label{lem:market-signal}
  In perfect Bayesian equilibrium of the Alice-Bob-Alice prediction market, without loss of generality, Alice's strategy is to predict $\vec{p}_S$ for some signaling scheme $\pi$ and associated random signal $S$.
\end{lemma}

\begin{proof}
  Let the random variable $S = \vec{p}^1$, i.e. Alice's report itself.
  In equilibrium, Bob observes $S$ and updates to posterior belief $\vec{p}_{S,B}$, reporting $\vec{p}^2 = \vec{p}_{S,B}$.
  Now consider the strategy profile where Alice reports $\vec{p}_S$ where she would have reported $S$, and otherwise strategies are unchanged.
  In this case, the total information available to Bob is still $(S,B)$, so he is still best-responding.
  Meanwhile, Alice still has the same information at round $3$ as Bob's strategy has not changed, so she is also still best-responding.
  So if the original strategy profile were an equilibrium, this profile is also an equilibrium, but one in which Alice receives strictly better utility. \qed
\end{proof}

Therefore, from here on we will describe Alice's strategy in prediction markets as a signaling scheme $\pi$, keeping in mind that she does not publicly announce her signal and does not have to commit to the scheme.


Before we proceed, we will give some necessary definitions.

\paragraph{Definitions.}
First, let us define $V = \Ex_{A,B,E} R(\vec{p}_{A,B},E) - R(\vec{p},E)$ where $\vec{p}$ is the prior.
This is the difference in expected score between the prior and the posterior conditioned on both players' signals (it can also be written $\Ex_{A,B} G(\vec{p}_{A,B}) - G(\vec{p})$).
Next, let us define the notation $u_B(\pi';\pi)$ as follows.
In the prediction market game, suppose Alice draws from $\pi$ while Bob believes she is drawing from $\pi'$.
If $\vec{p}^1$ is in the support of $\pi'$ given Bob's signal $B$, then he does a Bayesian update to an incorrect (in general) posterior belief $\vec{p}^2$ and reports it.
If $\vec{p}^1$ is not in the support of Alice's $\pi'$ strategy (``off the equilibrium path''), then Bob forms some belief over Alice's signal and uses this to again form an incorrect posterior belief $\vec{p}^2$.
We define $u_B(\pi';\pi)$ to be Bob's expected utility in this case, for some off-path beliefs of Bob.

\vspace{.1in}

The core idea occurs in the following lemma, which shows that, under some conditions, Alice prefers to deviate to the optimal signaling scheme.
\begin{lemma} \label{lem:alice-dev}
  Suppose that, in the ABA game with commitment, $\pi^*$ brings Alice higher utility than $\pi$.
  Then in the Alice-Bob-Alice prediction market, if Alice plays $\pi$ and always learns Bob's signal after his report, then Alice improves utility by deviating to $\pi^*$.
\end{lemma}

\begin{proof}

%
%

  Suppose in the prediction market that Alice plays $\pi$ and Bob best-responds.
  Suppose Bob's strategy reveals his signal, meaning that Alice is always able to provide the best-possible prediction $\mathbf{p}_{A,B}$.
  Then the total expected utility obtained by the players is $V$ as defined above.
  We note that Bob's utility will be $u_B(\pi;\pi)$.

  Meanwhile, in the ABA game with commitment where Alice plays $\pi^*$, Bob's expected utility is $u_B(\pi^*;\pi^*)$.
  As we have formulated it, the ABA game with commitment is a constant-sum game.
  So if $\pi^*$ is preferable to $\pi$ for Alice in that game, then
  \begin{equation}
    u_B(\pi^*;\pi^*) \leq u_B(\pi;\pi) . \label{eqn:optopt_leq_pipi}
  \end{equation}
  Now in the prediction market, suppose Alice deviates from $\pi$ to drawing $S$ according to $\pi^*$.
  Recall that, if Bob knew the true signaling scheme $\pi^*$ that Alice is using, he would respond with the true posterior $\vec{p}_{S,B}$.
  Let $\vec{w}_{S,B}$ be the prediction Bob actually makes when Alice reports according to $S$.
  Both on and off the equilibrium path, $\vec{w}_{S,B}$ is computed according to a Bayesian update according to the wrong signaling scheme, not the one Alice has actually deviated to.
  So, by strict properness of the scoring rule, Bob's utility satisfies
  \begin{align}
    u_B(\pi;\pi^*) &= \Ex_{B,S,E} R(\vec{w}_{S,B},E) - R(\vec{p}_S,E)  \nonumber \\
                   &< \Ex_{B,S,E} R(\vec{p}_{S,B},E) - R(\vec{p}_S,E)  \nonumber \\
                   &= u_B(\pi^*;\pi^*)   \label{eqn:piopt_leq_optopt}
  \end{align}
  where the inequality is due to strict properness of $R$. By combining Inequalities \eqref{eqn:optopt_leq_pipi} and \eqref{eqn:piopt_leq_optopt}, we get that Bob's expected utility is worse under this deviation by Alice.
  Because total expected utility is the constant $V$ under these conditions, Alice's expected utility is higher. \qed
\end{proof}

To prove our main result, we also need the following continuity claim.
\begin{lemma} \label{lem:payoffs-continuous}
  In the prediction market with differentiable $G$, fixing Bob's strategy, Alice's expected utility is continuous in $\pi$; and similarly, fixing Alice's strategy, Bob's expected utility is continuous with respect to each of his reports at the second stage (i.e. outcomes of $\vec{p}^2$) as well as each of the probabilities he places on each report.
\end{lemma}

\begin{proof}
  Fixing Bob's strategy, Alice's expected utility is simply $\sum_{a,s} \pi(s,a) u(s,a)$ where $u(s,a)$ is her expected utility conditioned on $S=s,A=a$.
  This is continuous in $\pi$.
  Fixing Alice's strategy, if Bob changes the probability of making a report, continuity follows for the same reason.
  If Bob changes a report $\vec{p}^2$ to $\vec{p'}^2$, his difference in expected score is a constant (the probability of making this report) times the difference $R(\vec{p}^2;p_{S,B}) - R(\vec{p'}^2;p_{S,B})$.
  By the characterization of \citet{savage1971elicitation}, $R$ is continuous in its first argument if derived from a differentiable convex function $G$.
  This follows because, according to that characterization, $R$ can be written as  $R(\vec{w}';\vec{w}) = G(\vec{w}') + \inprod{\nabla G(\vec{w}')} {\vec{w}-\vec{w}'}$, and differentiable convex functions are continuous and continuously differentiable. \qed
\end{proof}

These results allow us to prove the main result of this section.
\begin{theorem} \label{thm:scheme-eps-pbe}
  Let $\pi^*$ be the optimal signaling scheme for the ABA game with commitment, i.e. the minimizer of $u_B(\pi;\pi)$.
  Then for any $\epsilon$, there is an $\epsilon$-PBE of the Alice-Bob-Alice prediction market in which Alice plays within $\epsilon$ of $\pi^*$.
\end{theorem}

\begin{proof}
  First consider a modified Alice-Bob-Alice prediction market game with a finite, discretized report space for both players, i.e. a finite $\delta$-net for some $\delta$.
  Note that PBE exists in the discretized game because all report spaces are finite.
  Alice's report space is extended by adding the support of $\pi^*$ and $\mathbf{p}_{A,B}$.
 The game is also modified so that Bob's signal outcome $B=b$ is always announced publicly after his prediction is made in round two.
  In this game, Alice always learns $B$ at round two and plays $\vec{p}_{A,B}$ at round $3$ as a unique best response, by strict properness.
  So the total utility of the two players is $V$ and the game is constant-sum in expectation.
  Let $\pi^*$ minimize $u_B(\pi;\pi)$; then if Alice plays $\pi^*$ and Bob best-responds, his utility is minimized and by the constant-sum property, Alice's is maximized.
  Furthermore, by Lemma \ref{lem:alice-dev}, this is the only possible PBE, because for any other $\pi \neq \pi^*$, Alice has a profitable deviation by switching to $\pi^*$.
    
 Now suppose Bob continues playing from this strategy set in the original prediction market game, i.e. revealing his payoff.
  By continuity of payoffs (Lemma \ref{lem:payoffs-continuous}), he can do so while encoding $B$ in arbitrarily low-order bits for an arbitrarily small loss in expected utility.
  Bob loses at most, say, $\epsilon'$ utility for doing so, so it is an $\epsilon'$-equilibrium, proving the theorem.
  
Careful readers may raise an issue that in the original prediction market game, since Bob doesn't announce his signal outcome, Alice may not be able to learn $B$ at round two just from Bob's best-response prediction. This is indeed true in degenerate cases.\footnote{An example is when $A$ and $B$ are uniformly random bits and $E$ is their XOR.} However, this problem can be removed by a technique of \citet{kong2018optimizing}, which shows in Lemma 19 that Alice can modify $\pi^*$ arbitrarily slightly so that Bob's strict best-response reveals his signal. Again by continuity of the payoffs, we have that Alice loses only $\epsilon'$ by doing so. In non-degenerate cases, Alice can always infer Bob's signal from his report.
 
  \qed
\end{proof}

%% file: algs.tex
\section{ABA Game with Commitment is Bayesian Persuasion}\label{sec:aba-bp}


In this section, we formally establish the connection between the ABA game with commitment (denoted as $\texttt{ABA-Commit}$) and the Bayesian Persuasion (BP) game with a privately informed receiver (denoted as $\texttt{BP-Private}$). 
Besides revealing interesting conceptual insights, this connection also enables us to directly employ ideas from Bayesian persuasion to design an efficient algorithm for the ABA game  when the size of Bob's signal space is a constant and the expected score function $G$ is $k$-piecewise linear.

\subsection{Reducing $\texttt{ABA-Commit}$ to $\texttt{BP-Private}$}\label{sec:aba-bp:reduction}
 
 We start by simplifying the equilibrium analysis of the ABA game with commitment.  Since Bob has only one chance to participate in the ABA game, his optimal strategy is simply to reveal his original signal at $t=2$ (assuming tie breaking in favor of more information) and Alice will also reveal all her information at $t=3$. Therefore, the only non-trivial stage is Alice's optimal  commitment at the first stage. Since the game is constant-sum, so maximizing Alice's utility is equivalent to minimizing Bob's utility. As a result, solving the ABA game with commitment boils down to compute  \emph{Alice's optimal commitment} (to a signaling scheme) at the first stage to minimize Bob's utility. 
 
For convenience and clarity, we state the result for piecewise linear convex function $G$, however this connection holds for arbitrary convex $G$ function (see remarks at the end of the theorem proof). 
 
 \begin{theorem}\label{thm:aba-to-bp-reduction}
For any $\texttt{ABA-commit}$ instance $\{G, \mu \}$ where $G$ is $k$-piecewise linear and $\mu$ is the prior over $(A,B,E)$, there is a $\texttt{BP-private}$ instance such that Alice's optimal commitment  is the same as the sender's optimal commitment in the $\texttt{BP-private}$ instance, which is described as follows:  (1) the instance has the same joint prior $\mu$ over the sender signal $A$, receiver signal $B$ and event $E$; (2) The receiver utility function $U_G(i,e)$ is uniquely determined by $G$ with action set $[k] = \{  1, 2,\cdots,k\}$; (3) The sender utility as a function of any signaling scheme $\pi: \S \times \A \to [0,1]$ is given by
\begin{equation}\label{eq:bob-util}
\text{Sender Obj = }  \Ex_{s} \max_{ i \in [k]} \sum_{e \in E} [U_G(i,e) \cdot \Pr(e|s)]  - \Ex_{s,b} \max_{i \in [k]} \sum_{e \in E} [U_G(i,e) \cdot \Pr(e|s,b)]. 
\end{equation}
 
 \end{theorem}
 
\begin{proof}
One key difference between the ABA game and Bayesian persuasion is that the receiver in BP is a decision maker who takes an action whereas both of the two players in the ABA game are not decision makers. To relate the ABA game to the BP model, our key insight  is that the ``receiver'' (i.e., a decision maker) in the ABA game is neither Alice nor Bob; Instead, he is implicitly encoded in the expected score function $G$, as described in the following fact.


\begin{fact}\label{lem:equivalence}
	For any $k$-piecewise-linear convex function $G: \Delta(\E) \to \RR$, there exists a decision making problem $U(i,e)$ which depends on a decision maker's action $i \in [k]$ and a random event $e \in \E$, such that $G(\post{}) = \max_{i \in [k]} \Ex_{E \sim \post{}} U(i,E) = \max_{i \in [k]} \sum_{e \in E} p_{e} \cdot U(i,e) $ for all $\post{} \in \Delta(\E)$. 
	
	Conversely, for any decision making problem $U(i,e)$  for $i \in [k]$ and $e \in \E$, the decision maker's maximum expected utility $\max_{ i  \in [k]} \Ex_{E \sim \post{}} U(i,E)$ on belief $\post{}$ is a $k$-piecewise-linear convex function in $\post{}$. 
\end{fact} 

It is easy to verify the second part of the fact. To see that the first part is true, since $G: \Delta(\E) \to \RR$ is convex and $k$-piece-wise linear, we know there exist $k$ linear functions: $r^i \cdot \post{} + b^i$ for $i = 1,\cdots,k$ ($r^i \in \RR^{\E}, b^i \in \RR $), such that $G(\post{}) = \max_{i \in [k]} [ r^i \cdot \post{} + b^i  ]$. Since $\sum_{e \in \E} p_e = 1$, by letting $U(i,e) = r^i_e + b^i$, we have $G(\post{}) = \max_{i \in [k]} \Ex_{E \sim \post{}} U(i,E)$, as desired.

Fact \ref{lem:equivalence} illustrates that $k$-piecewise linear convex functions are in one-to-one correspondence to decision making problems with $k$ actions. For any such $G$, we use $U_G(i,e)$ to denote the payoff structure of the corresponding decision making problem.  This allows us to view the ABA game as the following Bayesian persuasion problem. The \emph{receiver} is a decision maker, who wants to take an action $i \in [k]$ (recall that $k$ is the number of pieces of $G$) with utility $U_G(i,e)$ where $e \in E$. Since under commitment, Bob always reveals all his information to the decision maker. This can be equivalently viewed as if the decision maker is directly, and privately, informed with  Bob's signal $B$. As a result, Alice's optimal commitment problem is precisely to persuade such a privately informed decision maker to minimize Bob's expected utility, or equivalently, maximize the negative of Bob's expected utility which is  $\Ex_{s} G(\post{s}) - \Ex_{sb} G(\post{sb}) $ where $s$ is a signal realization of Alice's signaling scheme $\pi: \S \times \A \to [0,1]$.  This completes our reduction from $\texttt{ABA-Commit}$ to $\texttt{BP-Private}$. We now derive the concrete form of the sender's objective function. 

Given signaling scheme $\pi: \S \times \A \to [0,1]$ such that $\pi(s,a) = \Pr[S=s, A=a]$, signal $s$ will be sent with probability $\sum_{a \in \A} \pi(s,a)$. Upon receiving signal $s$, the decision maker updates his belief about $a$, as follows: 
\begin{equation}
\Pr( a|s) = \frac{\pi(s,a)}{ \sum_{a \in \A} \pi(s,a)}
\end{equation}
and thus infers a posterior belief about event $e$ as 
\begin{equation}\label{eq:pr-e-s}
\Pr(e|s) = \sum_{a \in \A} \Pr(e|a) \cdot \Pr(a|s) =  \frac{1}{ \sum_{a \in \A}  \pi(s,a)} \sum_{a \in \A}  \mu(e|a) \cdot \pi(s,a). 
\end{equation}

Based on this belief, the decision maker will take an optimal action $\hat{i} = \arg \max_{i \in [k]} \sum_{e} [U_G(i,e) \cdot \Pr(e|s)]$. Note that  $\max_{i \in [k]} \sum_{e \in \E} [U_G(i,e) \cdot \Pr(e|s)]$ is precisely $G(\post{s})$, where $\post{s}(e) = \Pr(e|s)$. 

Now that Bob further reveals his signal $b$, then the decision maker infers a different posterior belief $\Pr(e|s,b)$ given by (see Appendix \ref{ap:proof-eq-pr-s-eb} for details):
\begin{equation}\label{eq:pr-e-sb}
\Pr(e|s,b) = \frac{ \sum_{a \in \A} \mu(e|a,b) \cdot \pi(s,a) \cdot \mu(b|a)}{ \sum_{a \in A} \pi(s,a) \cdot \mu(b|a)}. 
\end{equation}

Based on this belief, the decision maker will take an optimal action $\hat{i} = \arg \max_{i \in [k]} \sum_{e \in \E} [U_G(i,e) \cdot \Pr(e|s,b)]$. Note that  $\max_{i \in [k]} \sum_{e \in \E} [U_G(i,e) \cdot \Pr(e|s)]$ is precisely $G(\post{sb})$, where $\post{sb}(e) = \Pr(e|s,b)$. 

As a result, the sender's objective in our $\texttt{BP-private}$ instance (i.e., Alice's maximization objective) is the follows: 
\begin{equation}
\text{Sender Obj = }  \Ex_{s} \max_{ i \in [k]} \sum_{e \in E} [U_G(i,e) \cdot \Pr(e|s)]  - \Ex_{s,b} \max_{i \in [k]} \sum_{e \in E} [U_G(i,e) \cdot \Pr(e|s,b)]
\end{equation}
\qed
\end{proof}
 

 \begin{remark}
 The $k$-piecewise linear assumption in our reduction is only for clarity and notational convenience. The reduction does hold for general convex function $G$, in which case the receiver may need to pick an action from an infinite set. 
 We refer the reader to Appendix \ref{ap:persuasion-general} for more details. 
    
 \end{remark}

\subsection{A Direct Application of the Reduction}\label{subsec:k-linear-b-const-poly}

As a direction application of the reduction in Section \ref{sec:aba-bp:reduction}, we now show how to use this  connection to compute Alice's optimal commitment when $|\B|$ is constant and the expected score function $G$ is $k$-piecewise linear. Our algorithm is polynomial in $k$ but exponential in the constant $|\B|$, as described in the following theorem. 


\begin{theorem}\label{thm:k-linear-poly}
When $G$ is $k$-piecewise linear, there exists a $\emph{poly}(k^{|\B|}, |\A|, |\E|)$-time algorithm that computes Alice's optimal signaling scheme to commit to.
\end{theorem}

\begin{proofsketch}
See Appendix \ref{ap:proof-k-linear-poly} for the full proof; we give a sketch here.
First, we reframe the problem as a Bayesian persuasion problem with a privately informed receiver. Next, we prove the revelation principle adapted to our problem. We show that if two signals lead to the same decision-maker best-response under all values of Bob's signal $b \in \B$, then by merging the two signals, the decision maker's best response is the same. This is true because the constraints, expressed mathematically, are linear in the probabilities $\pi(s,a)$. This yields the following revelation principle: we can restrict attention to signaling schemes such that each signal is a set of obedient action recommendations, each corresponding to one possible value of private information (Bob's signal realization $b$).

 The optimal signaling scheme maximizes Alice's utility subject to the following incentive compatibility constraints:  for any action and signal realization, the receiver prefers the recommended action to any other action. Alice's utility and the incentive compatibility constraints depend on $\Pr(e|s)$ and $\Pr(e|s,b)$, and these posterior probabilities can be computed in terms of the prior $\mu$ and the signaling scheme $\pi(s,a)$. The resulting program is a linear program in $\pi$ with polynomially many variables and constraints, so it can be solved in polynomial time.
 \qed
\end{proofsketch}

In the introduction, we discussed the connection between $\texttt{ABA-commit}$ with informational substitutes and complements. Two signals are strong substitutes if the optimal signaling scheme is to always reveal all information, and two signals are strong complements if the optimal signaling scheme is to always reveal no information. We can use the algorithm in this section to compute the signaling scheme exactly. Therefore, the following corollary is immediate. 

\begin{corollary}
If $G$ is $k$-piecewise linear, then there exists a $\emph{poly}(k^{|\B|},|\A|,|\E|)$-time algorithm that tests whether two signals $A$ and $B$ are strong substitutes, complements, or neither.
\end{corollary}


\section{FPTAS for Different Parameter Regimes}\label{sec:fptas-regimes}
In this section, we develop Fully Polynomial Time Approximation Schemes (FPTAS) for the ABA game with commitment for different parameter regimes. These results cover a wider range of settings, and in particular, strictly generalize the main result of Kong and Schoenebeck \cite{kong2018optimizing}. Moreover, our algorithm is much simpler than that in \cite{kong2018optimizing} and is inspired by ideas that have also been used in the previous literature of algorithmic Bayesian persuasion.

While we do not use the explicit correspondence with the Bayesian persuasion instance developed in Section \ref{sec:aba-bp} here, we use key analytical techniques from the persuasion literature. Namely, the signaling scheme can be equivalently viewed as a distribution of posteriors and the only constraint on that distribution is the \emph{Bayes-plausibility} constraint: the expectation of the posteriors equal the prior. 
We then show that under a Lipschitz-like constraint on $G$, a small perturbation of the posterior leads to a small perturbation of Alice's payoff. We can therefore discretize the space of posteriors within $\epsilon$ precision and show that there exists an approximately optimal signaling mechanism whose induced posteriors lie only on those grid points. When the total number of grid points are polynomially bounded, we obtain efficient algorithms. This idea has been employed in algorithmic persuasion (e.g., \cite{cheng2015mixture,bhaskar2016hardness}). 



We start by defining the continuity condition we need on the expected score function $G$.

\begin{definition}[Local H\"{o}lder Continuity]
A function $G: \RR^n \to \RR$ is \emph{$(\alpha,\beta)$-locally H\"{o}lder continuous} if there exists $\alpha>0, \beta \in (0,1]$ and some $c \in (0,1)$ such that $|G(\vec{x}) - G(\vec{y})| \leq \alpha |\vec{x}-\vec{y}|^{\beta}$ for any $\vec{x},\vec{y}$ such that $|\vec{x}-\vec{y}| \leq c$.
\end{definition}


Note that local H\"{o}lder continuity  is a natural and weak continuity assumption, which holds for almost any reasonable scoring rule. In particular, it is weaker than the standard H\"{o}lder continuity, which requires the above condition to hold for any $\vec{x}, \vec{y}$, not only those with $|\vec{x}-\vec{y}| \leq c$.  H\"{o}lder continuity is then weaker than the Lipschitz continuity which corresponds to the  case of $\beta = 1$. 
Moreover, we will see later that $\alpha$ does not have to be an absolute constant; only that $\alpha$ is polynomial-sized is enough for an FPTAS. 

To obtain an FPTAS for the case with constant $|\A|$, Kong and Schoenebeck \cite{kong2018optimizing} defined another notion of continuity of  $G$, which they call \emph{niceness} condition formally described as follows. It turns out that \emph{niceness} condition  is a stronger requirement than  the local H\"{o}lder continuity.  So any function satisfying their condition also satisfies ours, including quadratic and log scoring rules.

\begin{definition}[Niceness Condition \cite{kong2018optimizing}]
	A function $G: \Delta_n \to \RR$ is $\lambda$\emph{-nice} if there exists a function $g: [0,1] \to \RR$ such that $G(\vec{x}) = \sum_{i} g(x_i)$ for every $\vec{x} \in \Delta_n$, $g(0)=g(1)=0$, $g$ is convex, and there exists a constant $\lambda \in (0,1)$ such that for sufficiently small $\epsilon$, $\max(|g(\epsilon),g(1-\epsilon)| \leq \epsilon^{\lambda}$.
\end{definition}

\begin{proposition}\label{prop:nice-implies-holder}
	Any function that is $\lambda$-nice for some $\lambda \leq 1$ is $(n^{1-\lambda},\lambda)$-locally H\"{o}lder continuous.\footnote{Note that if $\lambda > 1$ in the $\lambda$-nice condition, or if $\beta > 1$ in the $(\alpha,\beta)$-local H\"{o}lder continuity condition, then $G$ is identically zero so we are not interested in those trivial cases.}
\end{proposition}

The niceness condition is a relatively strong requirement, especially as requires the expected score function $G$ to be separable in all arguments $G(x) = \sum_{i} g(x_i)$. It happens to hold for log and quadratic scoring rules, but it is certainly not a property we generally expect to hold; the spherical scoring rule has $G(x) = (\sum_{i} x_i^2)^{1/2}$ which is not separable. 

\subsection{Constant Number of Alice's Signal Outcomes}\label{subsec:alice-constant}

We now consider the setting of \cite{kong2018optimizing} with constant size of Alice's signal space, i.e., $d \equiv |\A|$  is a constant. Kong and Schoenebeck \cite{kong2018optimizing}  prove that when $G$ satisfies the niceness condition, there is an FPTAS for this case. Here we exhibit another FPTAS for this setting based on the aforementioned idea from persuasion but under the (weaker) assumption of local H\"{o}lder continuity.  This thus strictly generalizes the result in \cite{kong2018optimizing}.  

Let $\Delta_{d} \equiv \Delta(\A)$ denote the set of all possible distributions over signal realizations of $A$. Let $\vec{p} \in \Delta_d$ denote a generic posterior distribution over Alice's signal space. Throughout we always use $|\vec{z}| = \sum_{i} |z_i|$ to denote the $l_1$ norm of a vector $\vec{z}$. For a function $f$, denote by $\{f(e)\}_{e \in \E}$ a vector of dimension $|\E|$ whose entries are $f(e)$ for $e \in \E$. We prove the following theorem, whose proof is deferred to Appendix \ref{ap:proof-a-const-apx}. 
\begin{theorem}\label{thm:a-const-apx}
Assume that $|\A|$ is a constant, and the $G$ function is  $(\alpha,\beta)$-locally H\"{o}lder continuous for some $\alpha, \beta > 0$ and bounded within $[-L,L]$ for some $L$. Then there exists a $\poly(|\B|,|\E|,1/\delta,L)$-time algorithm that computes Alice's $\delta$-optimal signaling scheme. 
\end{theorem}

\begin{proofsketch}
Recall that Alice's goal is to minimize Bob's expected utility. Let $\vec{w} \in \Delta(\A)$ be the posterior over Alice's signal space induced by her signal $s$. That is, $w_a = \Pr(a|s)$, where $w_a$ is the probability of $a \in \A$ assigned by $\vec{w}$. Let $u_B(\vec{w})$ denote Bob's utility as a function of Alice's report $\vec{w}$. We can do probability calculations to express $u_B(\vec{w})$ explicitly in terms of $\vec{w}$ and the prior $\mu$. Using this expression, we show that the value of $u_B$ does not change much if $\vec{w}$ does not change much in $l_1$ norm sense.
This is true because, from $\vec{w}$ to $\vec{w'}$, we can bound the absolute changes in expressions inside the $G(\cdot)$, and we can also bound the absolute changes in coefficients in front of $G()$, so triangle's inequality and the local H\"{o}lder continuity of $G$ allow us to conclude that the absolute change $|u_B(\vec{w})-u_B(\vec{w'})|$ is also bounded.

Now we define a $K$-uniform distribution to be a distribution whose entries are all multipliers of $1/K$, and let $\Delta_d(K)$ be the set of all $K$-uniform distributions. First, we show that we can unbiasedly approximate $\vec{w}$ by a distribution over this uniform grid if the grid is fine enough. More formally, for $K  \geq  \frac{\log(2d/\epsilon)d^2}{2\epsilon^2}$, there exists a distribution $\tilde{w}$ over $\Delta_d(K)$ such that $\mathbb{E}(\tilde{\vec{w}})=\vec{w}$, and $|\tilde{\vec{w}}-\vec{w}| \leq \epsilon$ with probability at least $1 - \epsilon$.  The result follows by letting $\tilde{\vec{w}}$ be an empirical average of $K$ samples from $\vec{w}$ and applying Hoeffding's inequality. 

We can then use this grid approximation result to prove the next key step, that there always exists an approximately optimal signaling scheme which is a decomposition over $K$-uniform distributions. This is true because the optimal signaling scheme is a distribution of posteriors, and for each posterior $\vec{w}_j$, we can replace it with its $\tilde{\vec{w}}$. Since $\vec{w}_j$ and $\tilde{\vec{w}}$ are close, $u_B(\vec{w}_j)$ and $\mathbb{E} u_B (\tilde{\vec{w}})$ are also close, and this is true for all $j \in \mathcal{J}$, so by replacing every posterior in the optimal signaling scheme with its uniform grid approximation, we get a signaling scheme that is approximately optimal 
whose posteriors are all $K$-uniform.

Lastly, there are only $\mathcal{O} (K^d)$ many posteriors that are $K$-uniform, so computing the approximately optimal signaling scheme in the previous paragraph reduces to solving an LP with one probability weight variable on each such posterior $\pi(\vec{w})$ for $\vec{w} \in \Delta_d(K)$, subject to Bayes plausibility, and the LP can be solved in polynomial time. We are done.

\qed
\end{proofsketch}

\subsection{Constant Number of Event Outcomes and Bob's Signal Outcomes}\label{subsec:event-bob-constant}

Next we exhibit an FPTAS for another parameter regime:  both $n_E \equiv |\E|$ and $n_B \equiv |\B|$ are constant. The proof uses the same technique as in the previous section, and can be found in Appendix \ref{ap:proof-eb-const-apx}. The key idea is that Alice's signaling scheme can be viewed equivalently as a distribution over posterior distributions $\mathbf{v} \in \Delta(\E \times \B)$ jointly over the event and the Bob's private signal, and that this distribution captures \textit{all} of the information needed.   Compared to Theorem \ref{thm:k-linear-poly}, this result does not require $k$-piecewise linearity of $G$ but requires that $|\E|$ is a constant. Moreover, this result is an FPTAS whereas Theorem \ref{thm:k-linear-poly} gives an exact algorithm. 

\begin{theorem}\label{thm:eb-const-apx}
Assume that $|\E|$ and $|\B|$ are constants, and the $G$ function is  $(\alpha,\beta)$-locally H\"{o}lder continuous for some $\alpha, \beta > 0$ and bounded within $[-L,L]$ for some $L$. Then there exists a $\poly(|\A|,1/\delta,L)$-time algorithm that computes Alice's $\delta$-optimal signaling scheme. 
\end{theorem}


%% file: conclusion.tex
\section{Conclusion and Directions} \label{sec:conclusion}
In this work, we took steps toward better understanding of equilibria of prediction markets, identifying informational substitutes and complements, and connections between these problems and other signaling games including Bayesian persuasion.

While these results extend the work of \cite{kong2018optimizing} in several ways -- connecting Alice's optimal commitment to the original prediction market game, generalizing results for the case of fixed $|\A|$, and new algorithms for other cases -- much open work still remains.
A first direction is to give efficient algorithms with fewer assumptions, e.g. if $|\B|$ is bounded but we have fewer restrictions on $G$.
It may be that persuasion-style techniques cannot be pushed much farther without additional structural results that are specific to the format of the prediction market game (as opposed to generic persuasion).

A second direction is to prove intractability results, which do not yet exist for this game,  although the problem appears quite challenging.
It would also be interesting to understand whether the problem of testing whether signals are informational substitutes is tractable or not, and whether computing Alice's optimal signaling scheme is algorithmically easier than testing substitutes. 

Finally, one can ask how these results extend to larger prediction market games.
In prior works on ``all-rush'' or ``all-delay'' equilibria~\cite{chen2010gaming,gao2013jointly,chen2016informational}, solving the Alice-Bob-Alice case tended to immediately extend to the general case of many players and trading periods.
However, when signals are neither substitutes nor complements but ``in between'', this extension is not clear.
Even computing the equilibrium of an Alice-Bob-Alice-Bob prediction market could require new backward-induction-style techniques.

%% file: persuasion-general.tex
\section{Omissions from Section \ref{sec:aba-bp} }

\subsection{Proof of Theorem \ref{thm:k-linear-poly}}\label{ap:proof-k-linear-poly}

We first give a proof outline. 
Fact \ref{lem:equivalence} allows us to reframe the problem as a Bayesian persuasion problem with privately informed receiver. Next, we prove the revelation principle (Lemma \ref{lem:revelation_principle}) adapted to our problem. The revelation principle states that we can restrict attention to signaling schemes such that each signal is a set of incentive compatible action recommendations, each corresponding to one possible value of private information (Bob's signal realization). The optimal signaling scheme maximizes Alice's utility subject to incentive compatibility constraints that for any action and signal realization, the receiver prefers the recommended action to any other action. Alice's utility and the incentive compatibility constraints depend on $\Pr(e|s)$ and $\Pr(e|s,b)$, and these posterior probabilities can be computed in terms of the prior $\mu$ and the signaling scheme $\pi(s,a)$. The resulting program is a linear program in $\pi$ with polynomially many variables and constraints, so it can be solved in polynomial time.\footnote{Bergemann and Morris \cite{Bergemann14Bayes} called such a signaling scheme \textit{Bayes Correlated Equilibrium} and showed that it can be computed by a linear program. Our argument and the linear program in the rest of this section is similar in spirit to theirs.} The rest of this section will carry out the outlined strategy in detail.


We start by proving a certain type of revelation principle for the ABA game with commitment. 
\begin{lemma}\label{lem:revelation_principle}[Revelation Principle]
	For any $k$-piece linear $G$, there always exists an optimal signaling scheme for Alice that uses at most $k^{|\B|+1}$ signals, with signal $s = \{ i_0\} \cup \{ i_b \}_{b \in B}$ resulting in action $i_0 \in [k]$ as the decision maker's best action when Bob does not reveal any signal and  $i_b \in [k]$ as the decision maker's best action when Bob reveals signal $b$.
\end{lemma} 
\begin{proof}
	Assume that there are two signals $s$ and $s'$ which result in the same decision maker best responses $i_0$  when Bob does not reveal any signal  and $i_b \in [k]$ for each Bob's signal $b$, we show that by merging signals $s,s'$ as one signal $\hat{s}$, $i_b$ is still the decision maker's best response action without seeing Bob's signal and $i_b$ is still the decision maker's best response action for Bob's signal $b$. Moreover, the decision maker's and Alice's utility will not change. 
	
	We first derive the conditions that signal $s = \{ i_0\} \cup \{ i_b \}_{b \in B}$ results in action $i_0$ as the decision maker's best action when Bob does not reveal any signal and  $i_b \in [k]$ as the decision maker's best action when Bob reveals signal $b$. This simply means  $i_0 = \arg \max_{i \in [k]} \sum_{e \in \E} [U_G(i,e) \cdot \Pr(e|s)]$ and $i_b = \arg \max_{i \in [k]} \sum_{e \in \E} [U_G(i,e) \cdot \Pr(e|s,b)]$ for all $b$. Mathematically, these can be formulated as the following constraints.
	
	\begin{align*}
		\sum_{e \in \E} [U_G(i_0,e) \cdot \Pr(e|s)]  &\geq \sum_{e \in \E} [U_G(i,e) \cdot \Pr(e|s)] \quad \forall i \in [k] \\
		\sum_{e \in \E} [U_G(i_b,e) \cdot \Pr(e|s,b)] &\geq \sum_{e \in \E} [U_G(i,e) \cdot \Pr(e|s,b)] \quad \forall i \in [k], b \in \B
	\end{align*}
	After substituting the expressions of $\Pr(e|s)$ and $\Pr(e|s,b)$ from (\ref{eq:pr-e-s}) and (\ref{eq:pr-e-sb}), the above constraints become the following linear constraints
	
	
	\begin{align*}
		\sum_{e \in \E, a \in \A} [U_G(i_0,e)  \mu(e|a)  \pi(s,a) ]  &\geq \sum_{e \in \E, a \in \A} [U_G(i,e)   \mu(e|a)  \pi(s,a) ] \quad \forall i \in [k] \\
		\sum_{e\in \E, a \in \A} [U_G(i_b,e)  \mu(e|a,b)  \pi(s,a)  \mu(b|a)  ]  &\geq 
		\sum_{e\in \E, a \in \A} [U_G(i,e)   \mu(e|a,b)  \pi(s,a)  \mu(b|a)  ]  \quad \forall i \in [k], b \in \B
	\end{align*}
	
	
	Crucially, these are all linear constraints of $\pi(s,a)$ for the fixed signal $s$. These constraints are also called \emph{obedience} or \emph{persuasiveness} constraints in the Bayesian Persuasion literature. So if $s, s'$ result in the same decision maker best responses in all the scenarios, by defining $\hat{s}$ with $\pi(\hat{s},a) = \pi(s,a) + \pi(s',a)$, $\hat{s}$ will result in the same decision best responses as $s,s'$ in all the scenarios. Moreover, it is easy to verify that the new scheme with $s,s'$ substituted by $\hat{s}$ will not change both the decision maker and Alice's expected utility.   
	\qed \end{proof}

Thanks to Lemma \ref{lem:revelation_principle}, we know that there exists an optimal signaling scheme for Alice which uses at most $k^{ |\B| + 1} $ signals, with signal $s = \{ i_0\} \cup \{ i_b \}_{b \in \B}$ resulting in action $i_0$ as the decision maker's best action when Bob does not reveal any signal and  $i_b \in [k]$ as the decision maker's best action when Bob reveals signal $b$. Let $\S$ denote the set of all these   $k^{ |\B| + 1} $ signals.
This lemma allows us to draw on the literature on Bayes correlated equilibria \cite{Bergemann14Bayes} to frame the problem as a linear program of $\text{ploy} (|\S|,k,|\A|)$ size, which we now derive. 

Alice's objective is to maximize the negative of Bob's utility, as follows:
$$ \Ex_{s}   \bigg[ \sum_{e \in \E} U_G(i_0,e) \cdot \Pr(e|s) \bigg] - \Ex_{s,b}  \bigg[ \sum_{e \in \E} U_G(i_b,e)  \cdot \Pr(e|s,b) \bigg]. $$
The obedience constraints for each signal $s \in \S$ are described as follows:
$$
\begin{aligned}
\sum_{e \in \E} [U_G(i_0,e) \cdot \Pr(e|s)]  &\geq \sum_{e \in \E} [U_G(i,e) \cdot \Pr(e|s)]  & \quad i \in [k],\\
\sum_{e \in \E} [U_G(i_b,e) \cdot \Pr(e|s,b)]  &\geq \sum_{e \in \E} [U_G(i,e) \cdot \Pr(e|s,b)]  & \quad i \in [k], b \in \B. 
\end{aligned}
$$
%

By substituting the expressions of $\Pr(e|s)$ and $\Pr(e|s,b)$ from (\ref{eq:pr-e-s}) and (\ref{eq:pr-e-sb}) in the above expressions, we can derive the following linear program for computing Alice's optimal commitment: 


\begin{small}
	\begin{gather}
	\begin{aligned}
		\max_{ \pi } \quad &  \sum_{s \in S,a \in \A, e \in \E}    \left[  U_G(i_0,e) \mu(e|a) \pi(s,a) - \sum_{b \in \B} U_G(i_b,e)  \mu(e, b|a)  \pi(s,a)  \right]   \quad & \\
		 \text{s.t.} \quad  & 
			\sum_{e \in \E, a \in \A} [U_G(i_0,e)  \mu(e|a)  \pi(s,a) ] \geq \sum_{e \in \E, a \in \A} [U_G(i,e)   \mu(e|a)  \pi(s,a) ]   & \forall  i \in [k], s \in \S, \\
		&	\sum_{e\in \E, a \in \A} [U_G(i_b,e)  \mu(e|a,b)  \pi(s,a)  \mu(b|a)  ]  \geq 
			\sum_{e\in \E, a \in \A} [U_G(i,e)   \mu(e|a,b)  \pi(s,a)  \mu(b|a)  ]  & \forall i \in [k], s \in \S, b \in \B ,  \\
		&	\sum_{s \in S} \pi(s,a)  = \mu(a)   & \forall a \in \A , \\
		&	\pi(s,a) \geq 0 &\quad \forall s \in \S, a \in \A . 
		\end{aligned}
	\end{gather}
\end{small}

This completes our proof of Theorem \ref{thm:k-linear-poly} since  $|\S| = k^{|\B|+1}$.

\subsection{Reducing $\texttt{ABA-Commit}$ to $\texttt{BP-Private}$ for General $G$} \label{ap:persuasion-general}

We use most of the notations of Section \ref{sec:aba-bp}, and let $G$ be any convex function. We think of $G$ as smooth (but it doesn't have to be; throughout we use the gradient of $G$, but we can use the subgradient of $G$ for general $G$ instead). 

If $G$ is $k$-piecewise linear, that is, it is a maximum of $k$ linear functions, then the decision maker (receiver) has $k$ actions and the action space is $[k]$. 

When $G$ is a general convex function, we can view $G$ as a maximum of infinitely many hyperplanes, and each hyperplane is a supporting hyperplane that is tangent to the graph of $G$ at each point $\vec{p} \in \Delta(\E)$. So we can have a decision maker whose action space is $\Delta(\E)$ and the utility to the decision maker of taking action $\vec{p} \in \Delta(\E)$ if the event is $e \in \E$ is
\begin{align*}
U_G(\vec{p},e) = G(\vec{p}) + \inprod{\nabla G(\vec{p})}{ \delta^{e} - \vec{p}}
\end{align*}
where $\nabla G(\vec{p})$ is the subgradient of $G$ evaluated at $\vec{p}$, and $\delta^{e} \in \Delta(\E)$ puts weight one on $e$ and zero elsewhere.

Alice's reported signal space is $\S = \Delta(E)^{B \cup \{0\}}$ and each $s \in \S$ can be written as $s = (s_0) \cup (s_b)_{b \in B}$, where $s_0 \in \Delta(\E)$ is a recommendation that the receiver/decision maker takes action $s_0$ when Bob reports nothing, and $s_b \in \Delta(\E)$ is a recommendation that the receiver/decision maker takes action $s_b$ if Bob's report is $b \in \B$. 

We can write a linear program analogously to that in Section \ref{sec:aba-bp} to characterize the optimal signaling scheme. The variables in the linear program are $\pi(s,a)$, a probability distribution over $\S \times \A = \Delta(\E)^{B \cup \{0\}} \times \A$. 

The optimal signaling scheme $\pi$ is a solution to the following (infinite-dimensional) LP.

\begin{align*}
\min_{\pi} &\int_{s \in \S} \sum_{a \in \A, e \in \E} \left[ \sum_{b \in \B} [U_G(s_b,e) \cdot \mu(e, b|a) \cdot \pi(s,a)] - U_G(s_0,e) \cdot \mu(e|a) \cdot \pi(s,a)] \right] \text{ s.t. } \\
s_0 &= \argmax_{\tilde{s} \in \Delta(\E)} \sum_{e \in \E, a \in \A} U_G(\tilde{s},e) \cdot \mu(e|a) \cdot \pi(s,a) \quad \forall s \in \S \\
s_b &= \argmax_{\tilde{s} \in \Delta(\E)} \sum_{e \in \E, a \in \A} U_G(\tilde{s},e) \cdot \mu(e|a,b) \cdot \pi(s,a) \cdot \mu(b|a) \quad \forall s \in \S, b \in \B \\
&\int_{s \in \S} \pi(s,a) = \mu(a) \quad \forall a \in \A \\
&\pi(s,a) \geq 0 \quad \forall s \in \S, a \in \A
\end{align*}
where $U_G(\vec{p},e) = G(\vec{p}) + \inprod{\nabla G(\vec{p})}{ \delta^{e} - \vec{p}}$ as stated above, and $s=(s_0, \{s_b\}_{b \in \B})$.

In Section \ref{sec:aba-bp}, the argmax is over $k$ discrete actions; here, the argmax is over $\Delta(\E)$, which is a compact space, so replacing the argmax with the first-order conditions is an instructive exercise. 
The first-order conditions are necessary conditions if the solutions are interior. 
If in addition $U_G(\vec{p},e)$ is convex in $\vec{p}$ for every $e \in \E$, then the first-order conditions are necessary and sufficient, so the new LP is equivalent to the old one.
If these conditions are satisfied, we can replace the argmax conditions with the first-order conditions and get the following equivalent LP.
\begin{align*}
\min_{\pi} &\int_{s \in \S} \sum_{a \in \A, e \in \E} \left[ \sum_{b \in \B} [U_G(s_b,e) \cdot \mu(e, b|a) \cdot \pi(s,a)] - U_G(s_0,e) \cdot \mu(e|a) \cdot \pi(s,a)] \right] \text{ s.t. } \\
&\sum_{e \in \E, a \in \A} \nabla U_G(s_0,e) \cdot \mu(e|a) \cdot \pi(s,a) \geq \sum_{e \in \E, a \in \A} \nabla U_G(\tilde{s},e) \cdot \mu(e|a) \cdot \pi(s,a)  \quad \forall s \in \S, \tilde{s} \in \Delta(\E) \\
&\sum_{e \in \E, a \in \A} \nabla U_G(s_b,e) \cdot \mu(e|a,b) \cdot \pi(s,a) \cdot \mu(b|a)  \\
& \qquad \geq \sum_{e \in \E, a \in \A} \nabla U_G(\tilde{s},e) \cdot \mu(e|a,b) \cdot \pi(s,a) \cdot \mu(b|a) \quad \forall s \in \S, b \in \B, \tilde{s} \in \Delta(\E) \\
&\int_{s \in \S} \pi(s,a) = \mu(a) \quad \forall a \in \A \\
&\pi(s,a) \geq 0 \quad \forall s \in \S, a \in \A
\end{align*}
%


\subsection{Omitted Probability Calculations}

\subsubsection{Proof of Equation \eqref{eq:pr-e-sb}}\label{ap:proof-eq-pr-s-eb}

\begin{eqnarray*}
	\Pr(e|s,b) &= & \sum_{a \in A}  \Pr(e|a,b) \cdot \Pr(a|s,b)  \\
	&=&  \sum_{a \in A} \mu(e|a,b) \cdot  \frac{\Pr(a,s,b)}{\Pr(s,b)} \\
	&=&  \sum_{a \in A} \mu(e|a,b) \cdot  \frac{\Pr(a,s)  \cdot \Pr(b|a,s)}{\Pr(s,b)} \\
	&=&  \sum_{a \in A} \mu(e|a,b) \cdot  \frac{\Pr(a,s)  \cdot \Pr(b|a,s)}{\sum_{a \in A}\Pr(a,s) \cdot \Pr(b|a,s)} \\
	&=&    \frac{ \sum_{a \in A} \mu(e|a,b) \cdot \pi(s,a) \cdot \mu(b|a)}{ \sum_{a \in A} \pi(s,a) \cdot \mu(b|a)} 
\end{eqnarray*} 

%% file: omitted-proofs.tex
\newpage 
\section{Omitted Proofs From Section \ref{sec:fptas-regimes} } \label{ap:omitted-proofs}

\subsection{Proof of Proposition \ref{prop:nice-implies-holder}}

	Assume that $G$ is $\lambda$-nice, then for $x_i,y_i \in \RR$ such that $x_i-y_i= \epsilon > 0$ is sufficiently small,
	\begin{align*}
		-\epsilon^{\lambda} = -|g(\epsilon)| = g(\epsilon)-g(0) \leq g(x_i)-g(y_i) \leq g(1) -g(1-\epsilon) = |g(1-\epsilon)| \leq \epsilon^{\lambda}
	\end{align*}
	where the inequalities come from the fact that $g$ is convex. We get analogous inequalities for $x_i < y_i$, so $|g(x_i)-g(y_i)| \leq |x_i-y_i|^{\lambda}$ for sufficiently small $|x_i-y_i|$.
	
	Let $\vec{x},\vec{y} \in \Delta_m$. If $|\vec{x}-\vec{y}|$ is sufficiently small, then so is $|x_i-y_i|$ for all $i$, so
	\begin{align*}
		\left|G(\vec{x})-G(\vec{y})\right|= \left|\sum_{i=1}^{n} g(x_i)-g(y_i)\right| \leq \sum_{i=1}^{n} |g(x_i)-g(y_i)| \leq \sum_{i=1}^{n} |x_i-y_i|^{\lambda} \leq n^{1-\lambda} \left(\sum_{i=1}^{n} |x_i-y_i| \right)^{\lambda}
	\end{align*}
	where the last inequality comes from H\"{o}lder's inequality.

\subsection{Proof of Theorem \ref{thm:a-const-apx} }\label{ap:proof-a-const-apx}

Recall that Alice's goal is to minimize Bob's expected utility. Let $\vec{w} \in \Delta(\A)$ be the posterior over Alice's signal space induced by her signal $s$. That is, $w_a = \Pr(a|s)$, where $w_a$ is the probability of $a \in \A$ assigned by $\vec{w}$. Let $u_B(\vec{w})$ denote Bob's utility as a function of Alice's report $\vec{w}$. The following lemma expresses $u_B(\vec{w})$ explicitly in terms of $\vec{w}$ and the prior $\mu$.
\begin{lemma}\label{lem:uB-w}
\begin{align*}
u_B(\vec{w}) = \sum_{b \in \B}  \left[\sum_{a \in \A} w_a \mu(b|a)  \right]  \times G\left( \left\{\frac{\sum_{a \in \A}  \mu(e|a,b)  w_a \mu(b|a) }{\sum_{a \in \A}  w_a \mu(b|a) } \right\}_{e \in \E} \right) -   G\left( \left\{\sum_{a \in \A} \mu(e|a) w_a \right\}_{e \in \E} \right)
\end{align*}
\end{lemma}

\begin{proof}
By definition of $u_B$, we have
\begin{align*}
u_B(\vec{w}) = \sum_{b \in \B} \Pr(b|s) G\big(\{\Pr(e|b,s)\}_{e \in \E}\big) - G\big(\{\Pr(e|s)\}_{e \in \E}\big)
\end{align*}

We then compute
\begin{align*}
\Pr(b|s) = \sum_{a \in \A} \Pr(a,b|s) = \sum_{a \in \A} \Pr(a|s) \Pr(b|a,s) = \sum_{a \in \A} w_a \mu(b|a) \\
\Pr(e|s) = \sum_{a \in \A} \Pr(e,a|s) = \sum_{a \in \A} \Pr(e|a,s)\Pr(a|s) = \sum_{a \in \A} \mu(e|a) w_a 
\end{align*}
Lastly,
\begin{align*}
\Pr(e|b,s) &= \frac{\Pr(e,b|s)}{\Pr(b|s)} = \frac{\sum_{a \in \A} \Pr(e,a,b|s)}{\sum_{a \in \A} \Pr(a,b|s)} \\ &= \frac{\sum_{a \in \A} \Pr(e|a,b,s)\Pr(a|s)\Pr(b|a,s) }{ \sum_{a \in \A} \Pr(a|s) \Pr(b|a,s) } = \frac{\sum_{a \in \A}  \mu(e|a,b) \cdot w_a \mu(b|a)}{ \sum_{a \in \A} w_a \mu(b|a) }
\end{align*}

These expressions immediately imply the lemma.
\qed \end{proof}


To prove the theorem, we first show that the value of $u_B$ does not change much if $\vec{w}$ does not change much in $l_1$ norm sense. 

\begin{lemma}\label{lem:G-continuous}
Assume the $G$ function is  $(\alpha,\beta)$-locally H\"{o}lder continuous for some $\alpha, \beta > 0$ and bounded within $[-L,L]$. Then we must have $|u_B(\vec{w}) - u_B(\vec{w}')| \leq 3|\B|\epsilon L + 3\alpha\epsilon^{1 - \beta} $ for any $\vec{w}, \vec{w}'$ such that $ |\vec{w}-\vec{w}'| \leq \frac{1}{2}  \epsilon^{1/\beta}$ and $\epsilon > 0$ sufficiently small. 
\end{lemma}
\begin{proof}
By $(\alpha,\beta)$-local H\"{o}lder continuity of $G$, we know that   $|G(\vec{x}) - G(\vec{y})| \leq \alpha|\vec{x}-\vec{y}|^{\beta}$ for any small enough $|\vec{x}-\vec{y}|$.  Now for any $\vec{w}' \in \Delta_n$ with $|\vec{w} - \vec{w}'| \leq \epsilon^{1/\beta}/2$, we will bound the difference between $u_B(\vec{w})$ and $u_B(\vec{w}')$. We start from the second term of $u_B(\vec{w})$ by bounding its input $\left\{ \sum_{a \in \A} \mu(e|a) w_a \right\}_{e \in \E} $. 

\begin{eqnarray*}
\left|\left\{ \sum_{a \in \A} \mu(e|a) w_a \right\}_{e \in \E}  - \left\{ \sum_{a \in \A} \mu(e|a) w'_a \right\}_{e \in \E} \right| =  \sum_{e \in \E} \sum_{a \in \A}  \mu(e|a) \left|w_a-w'_a\right| \\ =   \sum_{a \in \A} \sum_{e \in \E} \mu(e|a) \left|w_a-w'_a\right| = \sum_{a \in \A} \left|w_a-w'_a\right| = \left|\vec{w}-\vec{w}'\right|
\end{eqnarray*}

By $(\alpha,\beta)$-local H\"{o}lder continuity of $G$, we have
\begin{align*}
\left|G\left(\left\{ \sum_{a \in \A} \mu(e|a) w_a \right\}_{e \in \E} \right)  - G\left(\left\{ \sum_{a \in \A} \mu(e|a) w'_a \right\}_{e \in \E} \right)\right|  \leq \alpha |\vec{w}-\vec{w}'|^{\beta} \leq  \alpha \left( \frac{\epsilon^{1/\beta}}{2}  \right)^{\beta} \leq \alpha \epsilon \leq \alpha \epsilon^{1-\beta}
\end{align*}

Now we bound the first term. This turns out to be trickier.  For any fixed $b$, let $\lambda_b = \sum_{a \in \A} w_a \mu(b|a)$ and $\lambda'_b = \sum_a w_a' \mu(b|a)$. We have $$ \sum_{b \in \B} |\lambda_b - \lambda'_b| \leq \sum_{b \in \B} \sum_{a \in \A} \mu(b|a)  \left|w_a-w'_a\right| =  \sum_{a \in \A}  \left|w_a-w'_a\right| = \left|\vec{w} - \vec{w}'\right|.$$ 
Note that this also implies $|\lambda_b - \lambda'_b| \leq |\vec{w} - \vec{w}'|$. For any fixed $b$ such that $\lambda_b \geq \epsilon$, 

\begin{eqnarray*}
& &\left| \left\{ \frac{\sum_{a \in \A} \mu(e|a,b) w_a \mu(b|a) }{\sum_{a \in \A}  w_a \mu(b|a) } \right\}_{e \in \E}  - \left\{ \frac{\sum_{a \in \A}  \mu(e|a,b)  w_a' \mu(b|a) }{\sum_{a \in \A}  w_a' \mu(b|a) } \right\}_{e \in \E}  \right|  \\
& = & \left| \left\{ \sum_{a \in \A}  \mu(e|a,h) \mu(b|a) \left( \frac{w_a}{\lambda_b} - \frac{w_a'}{\lambda_b'}  \right) \right\}_{e \in \E}  \right| \\
& \leq & \sum_{a \in \A} \left| \left\{ \mu(e|a,b) \right\}_{e \in \E} \right| \mu(b|a) \left| \frac{w_a}{\lambda_b} - \frac{w_a'}{\lambda_b'} \right| \\
& = & \sum_{a \in \A}  \mu(b|a) \left| \frac{w_a}{\lambda_b} - \frac{w_a'}{\lambda_b'} \right| \\
& = & \sum_{a \in \A} \mu(b|a) \frac{|w_a\lambda_b' - w_a' \lambda_b|}{\lambda_b \lambda_b'} \\
& \leq & \sum_{a \in \A} \mu(b|a) \frac{|w_a-w_a'| \cdot \lambda_b + |\lambda_b - \lambda_b'| \cdot w_a}{\lambda_b \lambda_b'} \\
& \leq & \sum_{a \in \A} \frac{|w_a-w_a'|}{\lambda_b'} + \frac{\sum_{a \in \A} \mu(b|a) w_a \cdot |\lambda_b - \lambda_b'|}{\lambda_b \lambda_b'} \\
& = & \frac{|\vec{w}-\vec{w}'|}{\lambda_b'} + \frac{|\lambda_b-\lambda_b'|}{\lambda_b'} \leq \frac{2|\vec{w}-\vec{w}'|}{\lambda_b'} \leq \frac{2(\epsilon^{1/\beta}/2)}{\epsilon-\epsilon^{1/\beta}/2} \leq \frac{\epsilon^{1/\beta}}{\epsilon-\epsilon/2} = 2\epsilon^{1/\beta-1}
\end{eqnarray*}

where the second inequality used the fact that $|ab - cd| = |(a-c)d+(b-d)a| \leq |a - c| d + |b-d| a $. 

We are now ready to bound the difference of the first term of $u_B$, as follows.
\begin{eqnarray*}
& &  \sum_{b \in \B}   \lambda_b   \times G\left( \left\{ \frac{\sum_{a \in \A}  \mu(e|a,b) w_a \mu(b|a) }{\lambda_b} \right\}_{e \in \E} \right) - \sum_{b \in \B}   \lambda'_b   \times G\left( \left\{ \frac{\sum_{a \in \A}  \mu(e|a,b) w'_a \mu(b|a) }{\lambda'_b}
\right\}_{e \in \E} \right)  \\
& \leq & \sum_{b \in \B}   |\lambda_b - \lambda'_b| G\left( \left\{ \frac{\sum_{a \in \A}  \mu(e|a,b) w_a \mu(b|a) }{\lambda'_b} \right\}_{e \in \E} \right)   \\
& & + \sum_{b \in \B} \lambda_b \cdot \left| G\left( \left\{ \frac{\sum_{a \in \A}  \mu(e|a,b) w_a \mu(b|a) }{\lambda_b} \right\}_{e \in \E} \right) - G\left( \left\{ \frac{\sum_{a \in \A}  \mu(e|a,b) w_a' \mu(b|a) }{\lambda'_b} \right\}_{e \in \E} \right)  \right| \\ 
&\leq & |\B| |\vec{w}  - \vec{w}'| L + \sum_{b: \lambda_b \leq \epsilon } \lambda_b \cdot 2  L +  \sum_{b: \lambda_b \geq \epsilon}    \lambda_b \cdot  \alpha \bigg( 2 \epsilon^{1/\beta-1} \bigg)^{\beta}  \\
&\leq & |\B| \epsilon^{1/\beta} L/2 + 2|\B| \epsilon L + 2\alpha \epsilon^{1-\beta} \\
& \leq & 3|\B| \epsilon L + 2\alpha\epsilon^{1 - \beta}  
\end{eqnarray*}

Earlier we see that the difference of the second term of $u_B$ is bounded above by $\alpha \epsilon^{1-\beta}$. Combining the two finishes the proof.
\qed \end{proof}

\begin{corollary}\label{cor:G-continuous-epsdelta}
Assume conditions in Lemma \ref{lem:G-continuous}. For any $\delta>0$,  let \begin{equation}\label{eq:eps-delta}
\epsilon = \min \left\{ \frac{1}{2}\left(\frac{\delta}{6 |\B| L}\right)^{1/\beta}, \frac{1}{2} \left( \frac{\delta}{6\alpha} \right)^{1/\beta(1-\beta)} \right\}, 
\end{equation}
then we have $|u_B(\vec{w}) - u_B(\vec{w}')| \leq \delta$ for any $|\vec{w} - \vec{w}'| \leq \epsilon$. 
\end{corollary}
\begin{proof}
In Lemma \ref{lem:G-continuous}, choose $\epsilon$ such that  $3|\B| \epsilon L \leq \delta/2$ and $3\alpha \epsilon^{1-\beta} \leq \delta/2$, then map $\epsilon$ to $\epsilon^{1/\beta}/2$.  
\qed \end{proof}

We now show that there always exists an approximately optimal signaling scheme which is a decomposition over $K$-\emph{uniform distributions}, defined as follows.

\begin{definition}[$K$-uniform distributions]
	Any $\vec{w} \in \Delta_d$ is called a $K$-\emph{uniform distribution} if each entry of $\vec{w}$ is a multiplier of $1/K$. Let $\Delta_d(K) \subseteq \Delta_d$ denote the set of all $K$-uniform distributions in $\Delta_d$. 
\end{definition}

\begin{lemma}\label{lem:uniform-decompose}
	For any $K \geq \frac{\log(2d/\epsilon)d^2}{2\epsilon^2}$, there exists a distribution $\tilde{\vec{w}}$ over $\Delta_d(K)$ such that $\Ex(\tilde{\vec{w}}) = \vec{w}$ and $\Pr(|\tilde{\vec{w}} - \vec{w}| \geq \epsilon) \leq \epsilon$. 
\end{lemma}
\begin{proof}
	We take $K$ samples from distribution $\vec{w}$ and let $\tilde{\vec{w}}$ be empirical distribution over these $K$ samples. Note that $\tilde{\vec{w}}$ is a $K$-uniform distribution. Moreover, $\tilde{\vec{w}}$ can also be viewed as a random variable supported on $\Delta_d(K)$ (randomness comes from the sampling) with mean equaling precisely $\vec{w}$, i.e., $\Ex(\tilde{\vec{w}}) = \vec{w}$. Moreover, by Hoeffding's bound, we have $
	\Pr(| \tilde{w}_i - w_i| \geq \epsilon/d) \leq 2e^{-2k (\epsilon/d)^2} 
	$ for each $i \in [d]$. Therefore, the union bound implies
\begin{align*}
\Pr(|\tilde{\vec{w}} - \vec{w}| \geq  \epsilon) \leq \sum_{i=1}^d \Pr(| \tilde{w}_i - w_i| \geq \epsilon/d) \leq 2de^{-2K (\epsilon/d)^2}
\end{align*}
Let $K = \frac{\log(2d/\epsilon)d^2}{2\epsilon^2}$, we have  $2de^{-2K (\epsilon/d)^2}  \leq \epsilon$ as desired. 
\qed \end{proof}

\begin{lemma}\label{lem:exist-coarse-signals}
For any $\delta > 0$, let $\epsilon$ be as defined in Equation \eqref{eq:eps-delta} and $K = \frac{ \log(2d/\epsilon) d^2}{2\epsilon^2}$. There always exists a $(4L \epsilon + \delta)$-optimal signaling scheme whose posterior beliefs are all $K$-uniform (i.e., in  $\Delta_d(K)$).  
\end{lemma}
\begin{proof}

Let $\{ \lambda_j, \vec{w}_j \}_{j \in \mathcal{J}} $ be the optimal signaling scheme where posterior $\vec{w}_j$ is induced with probability $\lambda_j$, for each $j \in \mathcal{J}$ where $\mathcal{J}$ is an index set.\footnote{From \cite{kong2018optimizing} we can take $|\mathcal{J}| \leq d$ but we do not need this fact.} By Lemma \ref{lem:uniform-decompose}, we know that any $\vec{w}_j$, there exists a distribution $\tilde{\vec{w}}$ over $\Delta_n(K)$ such that $\Ex(\tilde{\vec{w}}) = \vec{w}_j$ and $\Pr(|\tilde{\vec{w}} - \vec{w}| \geq \epsilon) \leq \epsilon$. As a result, if we substitute any posterior $\vec{w}_j$ be the $\tilde{\vec{w}}$, Bob's utility change is upper bounded as follows:
\begin{eqnarray*}
&& |u_B(\vec{w}_j) - \Ex_{\tilde{\vec{w}}} u_B(\tilde{\vec{w}})| \\
& \leq &  |u_B(\vec{w}_i) - \Ex_{ \tilde{\vec{w}}:  |\tilde{\vec{w}} - \vec{w}_i| \geq \epsilon} u_B(\tilde{\vec{w}})| \cdot \Pr( |\tilde{\vec{w}} - \vec{w}_j| \geq \epsilon ) +  |u_B(\vec{w}_i) - \Ex_{ \tilde{\vec{w}}:  |\tilde{\vec{w}} - \vec{w}_j| \leq \epsilon} u_B(\tilde{\vec{w}})| \cdot \Pr( |\tilde{\vec{w}} - \vec{w}_j| \leq \epsilon ) \\
& \leq & 4L \epsilon + \delta 
\end{eqnarray*}  
where we used the fact that $u_B \leq 2 L$. Therefore, if we substitute all the $\vec{w}_j$'s by the corresponding $\tilde{\vec{w}}$, Bob's utility change is also bounded by $4L \epsilon + \delta$. In other words, there exists an $(4L \epsilon + \delta)$-optimal signaling scheme whose posteriors are all $K$-uniform. 
\qed \end{proof}

As a result of Lemma \ref{lem:exist-coarse-signals}, the following LP computes the optimal signaling schemes with posteriors from $\Delta_d(K)$, thus outputs a $(4L \epsilon + \delta)$-optimal signaling scheme. Since $\Delta_d(K)$ has $\poly(|\B|,1/\delta,L)$  elements when $d$ is a constant, this is a $\poly(|\B|,|\E|,1/\delta,L)$ time algorithm. 

\begin{eqnarray*}
	\min_{\{\pi(\vec{w})\}_{\vec{w} \in \Delta_{d}(K)}}  && \sum_{ \vec{w} \in \Delta_{d}(K) } u_B(\vec{w}) \cdot \pi(\vec{w})  \\
	\text{s.t.} && \sum_{\vec{w} \in \Delta_{d}(K)}  \vec{w} \cdot \pi(\vec{w})  = \{\mu(a)\}_{a \in \A} \\
	&& \sum_{\vec{w} \in \Delta_{d}(K)} \pi(\vec{w})  = 1\\
	&& \pi(\vec{w}) \geq 0 \quad \forall \vec{w} \in \Delta_d(K)
\end{eqnarray*}

\subsection{Proof of Theorem \ref{thm:eb-const-apx}}\label{ap:proof-eb-const-apx}

Let $\vec{v} \in \Delta(\E \times \B)$ be the posterior distribution over $\E \times \B$ after Alice's report, that is, if Alice's signal is $s$, then $v_{e,b} = \Pr(e,b|s)$ for $e \in \E, b \in \B$. The following lemma gives Bob's utility $u_B(\vec{v})$ explicitly as a function of $\vec{v}$ and the prior $\mu$.

\begin{lemma}
\begin{align*}
u_B(\vec{v}) = \sum_{b \in \B} \left[ \sum_{e \in \E} v_{e,b} \right] G \left(  \left\{ \frac{v_{e,b}}{\sum_{\tilde{e} \in \E}  v_{\tilde{e},b}} \right\}_{e \in \E} \right) - G \left( \left\{ \sum_{b \in \B} v_{e,b} \right\}_{e \in \E} \right)
\end{align*}
\end{lemma}

\begin{proof}
By definition of $u_B$, we have
\begin{align*}
u_B(\vec{v}) = \sum_{b \in \B} \Pr(b|s) G\left(\left\{\Pr(e|b,s)\right\}_{e \in \E}\right) - G\left(\left\{\Pr(e|s)\right\}_{e \in \E}\right)
\end{align*}
where $\left\{\Pr(e|b,s)\right\}_{e \in \E} \in \Delta(\E)$ is a vector whose entries are $\Pr(e|s,b)$ for $e \in \E$ and analogously for $\left\{\Pr(e|s)\right\}_{e \in \E}$.
We then compute
\begin{align*}
\Pr(b|s) = \sum_{e \in \E} \Pr(e,b|s) = \sum_{e \in \E} v_{e,b} \\
\Pr(e|s) = \sum_{b \in \B} \Pr(e,b|s) = \sum_{b \in \B} v_{e,b} \\
\Pr(e|b,s) = \frac{\Pr(e,b|s)}{\Pr(b|s)} = \frac{v_{e,b}}{\sum_{\tilde{e} \in \E}  v_{\tilde{e},b}}
\end{align*}
These expressions immediately imply the lemma.
\qed \end{proof}


We also use the $l_1$ norm on vectors: $|\vec{v}-\vec{v}'| = \sum_{e \in \E} \sum_{b \in \B} |v_{e,b}-v'_{e,b}|$.

\begin{lemma}\label{lem:G-continuous-nEnB}
Assume the $G$ function is  $(\alpha,\beta)$-locally H\"{o}lder continuous for some $\alpha, \beta > 0$, and bounded within $[-L,L]$. Then we must have $|u_B(\vec{v}) - u_B(\vec{v}')| \leq 3|\B|\epsilon L + 3\alpha\epsilon^{1 - \beta}  $ for any $\vec{v}, \vec{v}'$ such that $ |\vec{v}-\vec{v}'| \leq \frac{1}{2}  \epsilon^{1/\beta}$. 
\end{lemma}

\begin{proof}
For each $b \in \B$, let $\lambda_b = \sum_{e \in \E} v_{e,b}$ and $\lambda'_b = \sum_{e \in \E} v'_{e,b}$
We first note that 
\begin{align*}
|\lambda_b-\lambda'_b|  = |\sum_{e \in \E} v_{e,b}-v'_{e,b}| \leq \sum_{e \in \E} |v_{e,b}-v'_{e,b}| = |\{v_{e,b}\}_{e \in \E}-\{v'_{e,b}\}_{e \in \E}| \\ \leq \sum_{b \in \B} \sum_{e \in \E} |v_{e,b}-v'_{e,b}| = |\vec{v}-\vec{v}'| \leq \frac{1}{2}\epsilon^{1/\beta}
\end{align*}
for every fixed $b \in \B$.

We first bound the second term. We have
\begin{align*}
\left|G\left(\{v_{e,b}\}_{e \in \E} \right)-G\left(\{v'_{e,b}\}_{e \in \E} \right) \right| \leq \alpha \left|\{v_{e,b}\}_{e \in \E}-\{v'_{e,b}\}_{e \in \E} \right|^{\beta} \leq  \alpha \left( \frac{\epsilon^{1/\beta}}{2}  \right)^{\beta} \leq \alpha \epsilon \leq \alpha \epsilon^{1-\beta}
\end{align*}

Now we bound the first term. For any fixed $b$ such that $q(b) \geq \epsilon$,

\begin{align*}
 \left| \left\{ \frac{v_{e,b}}{\lambda_b} \right\}_{e \in \E} - \left\{ \frac{v'_{e,b}}{\lambda'_b} \right\}_{e \in \E} \right| &= \sum_{e \in \E} \left| \frac{v_{e,b}}{\lambda_b} - \frac{v'_{e,b}}{\lambda'_b} \right| \\
&= \sum_{e \in \E} \frac{|v_{e,b}\lambda'_b-v'_{e,b}\lambda_b|}{\lambda_b \lambda'_b} \\
&\leq \sum_{e \in \E} \frac{|v_{e,b}-v'_{e,b}|\lambda_b + |\lambda'_b-\lambda_b|v_{e,b}}{\lambda_b \lambda'_b} \\
&= \frac{1}{\lambda_b} \sum_{e \in \E} |v_{e,b}-v'_{e,b}| + \frac{|\lambda'_b-\lambda_b|}{\lambda_b \lambda'_b} \lambda_b \\
&\leq \left( \frac{1}{\lambda_b} + \frac{1}{\lambda'_b} \right) |\vec{v}-\vec{v}'| \\
&\leq \left( \frac{1}{\epsilon} + \frac{1}{\epsilon-\epsilon^{1/\beta}/2} \right) \frac{\epsilon^{1/\beta}}{2} \\ &\leq \frac{\epsilon^{1/\beta}}{\epsilon-\epsilon^{1/\beta}/2} \leq \frac{\epsilon^{1/\beta}}{\epsilon-\epsilon/2} = 2 \epsilon^{1/\beta-1}
\end{align*}

Earlier we see that the difference of the second term of $u_B$ is bounded above by $\alpha \epsilon^{1-\beta}$. Combining the two finishes the proof.

\qed \end{proof}

\begin{corollary}\label{cor:G-continuous-nEnB-epsdelta}
Assume conditions in Lemma \ref{lem:G-continuous-nEnB}a. For any $\delta>0$,  let \begin{equation}\label{eq:eps-delta-nEnB}
\epsilon = \min \left\{ \frac{1}{2}\left(\frac{\delta}{6n_BL}\right)^{1/\beta}, \frac{1}{2} \left( \frac{\delta}{6\alpha} \right)^{1/\beta(1-\beta)} \right\}, 
\end{equation}
then we have $|u_B(\vec{v}) - u_B(\vec{v}')| \leq \delta$ for any $|\vec{v} - \vec{v}'| \leq \epsilon$. 
\end{corollary}
\begin{proof}
The proof is completely analogous to the proof of Lemma \ref{cor:G-continuous-epsdelta}.  
\qed \end{proof}

\begin{lemma}\label{lem:uniform-decompose-nEnB}
	For any $K \geq \frac{\log(2(n_E+n_B)/\epsilon)(n_E+n_B)^2}{2\epsilon^2}$, there exists a distribution $\tilde{\vec{v}}$ of $(E,B)$ over $\Delta_{n_E+n_B}(K)$ such that $\Ex(\tilde{\vec{v}}) = \vec{v}$ and $\Pr(|\tilde{\vec{v}} - \vec{v}| \geq \epsilon) \leq \epsilon$. 
\end{lemma}
\begin{proof}
The proof is completely analogous to the proof of Lemma \ref{lem:uniform-decompose}.
\qed \end{proof}

\begin{lemma}\label{lem:exist-coarse-signals-nEnB}
For any $\delta > 0$, let $\epsilon$ be as defined in Equation \eqref{eq:eps-delta-nEnB} and $K = \frac{ \log(2(n_E+n_B)/\epsilon) (n_E+n_B)^2}{2\epsilon^2}$. There always exists a $(4L \epsilon + \delta)$-optimal signaling scheme whose posterior beliefs over $(E,B)$ are all $K$-uniform (i.e., in  $\Delta_{n_E+n_B}(K)$).  
\end{lemma}
\begin{proof}
The proof is completely analogous to the proof of Lemma \ref{lem:exist-coarse-signals}.
\qed \end{proof}

As a result of Lemma \ref{lem:exist-coarse-signals-nEnB}, the following LP computes the optimal signaling schemes with posteriors from $\Delta_{n_E+n_B}(K)$, thus outputs a $(4L \epsilon + \delta)$-optimal signaling scheme. Because $n_E$ and $n_B$ are constants, $1/\epsilon = \poly(1/\delta,L)$ (see (\ref{eq:eps-delta-nEnB})), so $\Delta_{n_E+n_B}(K)$ has $\poly(1/\delta,L)$ elements when $n_E$ and $n_B$ are constants. Therefore, solving this LP is a $\poly(1/\delta,|\A|,L)$-time algorithm. 

\begin{eqnarray*}
	\min_{ \{\pi(\vec{v})\}_{\vec{v} \in \Delta_{n_E+n_B}(K)}  }  && \sum_{ \vec{v} \in \Delta_{n_E+n_B}(K) } u_B(\vec{v}) \cdot \pi(\vec{v})  \\
	\text{s.t.} && \sum_{\vec{v} \in \Delta_{n_E+n_B}(K)}  \vec{v} \cdot \pi(\vec{v})  = \{ \{ \mu(e) \}_{e \in \E}, \{ \mu(b) \}_{b \in \B} \} \\
	&& \sum_{\vec{v} \in \Delta_{n_E+n_B}(K)} \pi(\vec{v})  = 1\\
	&& \pi(\vec{v}) \geq 0 \quad \forall \vec{v} \in \Delta_{n_E+n_B}(K)
\end{eqnarray*}

%